\documentclass[preprint,11pt]{elsarticle}

\makeatletter
\def\ps@pprintTitle{%
 \let\@oddhead\@empty
 \let\@evenhead\@empty
 \def\@oddfoot{}%
 \let\@evenfoot\@oddfoot}
\makeatother

\usepackage[titletoc,title]{appendix}

\makeatletter
  \@addtoreset{chapter}{part}
  \@addtoreset{@ppsaveapp}{part}
\makeatother

\usepackage{amsmath,amsthm}
\usepackage{relsize}
\usepackage{algorithm}
\usepackage{algcompatible}
\usepackage{setspace}
\floatname{algorithm}{}
\let\Algorithm\algorithm
\renewcommand\algorithm[1][]{\Algorithm[#1]\setstretch{1.2}}
\renewcommand{\thealgorithm}{}
\usepackage{euscript}
\usepackage{enumerate}
\usepackage{enumitem}
\setlist{leftmargin=5.5mm}
\usepackage{subfigure}
\usepackage{tikz}
\usetikzlibrary{calc}
\usepackage{verbatim} 
\usepackage{xspace} 
\usetikzlibrary{decorations.pathreplacing}

\usetikzlibrary{arrows,decorations.pathmorphing,backgrounds,positioning,fit,petri}
\newtheorem{theorem}{Theorem}
\newtheorem{prop}[theorem]{Proposition}

\newtheorem{coro}[theorem]{Corollary}

\newtheorem{lem}[theorem]{Lemma}

\def\eps{\varepsilon}

\newcommand{\KPone}{$Sol(KP_1)$\xspace}
\newcommand{\KPsec}{$Sol(KP_2)$\xspace}

\newcommand{\KPt}{$Sol(KP_t)$\xspace}

\newcommand{\IIKP}{$\mbox{IIKP}$\xspace}

\begin{document}
\pagestyle{plain}

\begin{frontmatter}

\title{Approximating the Incremental Knapsack Problem}

\date{}

 \author[1,2]{Federico Della Croce}
 \author[3]{Ulrich Pferschy}
 \author[1]{Rosario Scatamacchia\corref{cor1}}
\cortext[cor1]{Corresponding author.}
 
  \address[1]{\small Dipartimento di Ingegneria Gestionale e della Produzione, Politecnico di Torino,\\ Corso Duca degli Abruzzi 24, 10129 Torino, Italy, \\{\tt \{federico.dellacroce, rosario.scatamacchia\}@polito.it }}
 \address[2]{CNR, IEIIT, Torino, Italy}
 \address[3]{\small Department of Statistics and Operations Research, University of Graz, Universitaetsstrasse 15, 8010 Graz, Austria, \\ {\tt pferschy@uni-graz.at}}

\begin{abstract}
We consider the 0--1 Incremental Knapsack Problem (IKP) where the capacity grows over time periods and if an item is placed in the knapsack in a certain period, it cannot be removed afterwards. The contribution of a packed item in each time period depends on its profit as well as on a time factor which reflects the importance of the period in the objective function. The problem calls for maximizing the weighted sum of the profits over the whole time horizon. 
In this work, we provide approximation results for IKP and its restricted variants. In some results, we rely on Linear Programming (LP) to derive approximation bounds and show how the proposed LP--based analysis can be seen as a valid alternative to more formal proof systems. We first manage to prove the tightness of some approximation ratios of a general purpose algorithm currently available in the literature and originally applied to a time-invariant version of the problem. We also devise a Polynomial Time Approximation Scheme (PTAS) when the input value indicating the number of periods is considered as a constant. Then, we add the mild and natural assumption that each item can be packed in the first time period. For this variant, we discuss different approximation algorithms suited for any number of time periods and for the special case with two periods. 
\end{abstract}
\begin{keyword} Incremental Knapsack problem \sep Approximation scheme \sep Linear Programming 
\end{keyword}

\end{frontmatter}

\section{Introduction}
\label{sec:introIKP}

The 0--1 Knapsack Problem (KP) is one of the paradigmatic problems in combinatorial optimization and has been object of numerous publications 
and two monographs 
\cite{MarTot90} and \cite{KePfPi04}.
In KP a set of items with given profits and weights is available and the aim is to select a subset of the items in order to maximize the total profit without exceeding a known knapsack capacity.
KP is weakly NP-hard, although in practice fairly large instances can be solved to optimality within
limited running time.
Various generalizations of KP have been considered over the years. 
We mention among others the most recent contributions we are aware of:
the collapsing KP (where the capacity of the constraint is  inversely related to the number of items placed inside the knapsack) \cite{CDSS17};
the discounted KP (where it is required to select a set of item groups where each group includes three items and at most one of the three items can be selected) \cite{CXYSW16};
the parametric KP (where the profits of the items are affine-linear functions of a real-valued parameter and the task is to compute a solution for all values of the parameter) \cite{GHRT17};
the KP with setups (where items belong to disjoint families 
and can be selected only if the corresponding family is activated) 
\cite{SDSS17, FUMOTR18, PS17}; the penalized KP (where each item has a profit, a weight, a penalty and the goal is to maximize the sum of the profits minus the greatest penalty value of the selected items) \cite{DCPFSC17} and 
the temporal KP (where a time horizon is considered, and each item consumes the knapsack capacity during a limited time interval only) \cite{CFMT16}.

We consider here a very meaningful and natural generalization of KP, namely the 0--1 Incremental Knapsack Problem ($\mbox{IKP}$) as introduced in \cite{BiSeYe13} where the constraining capacity grows over $T$ time periods. 
If an item is placed in the knapsack in a certain period, it cannot be removed afterwards. 
Each packed item contributes its profit for each time period in which it is included in the knapsack, 
with different impacts in the objective function depending on multiplicative time factors.
These reflect the different importance of the periods or allow a discounting over time.
The problem calls for maximizing the weighted sum of the profits accumulated over the whole time horizon. \\
$\mbox{IKP}$ generalizes the time-invariant Incremental Knapsack Problem, namely with unit time multipliers, considered in \cite{hart08, HaSh06, sha07} and which we will denote here as \IIKP.
$\mbox{IKP}$ has many real-life applications since, from a practical perspective, it is often required in resource allocation problems to deal with changes in the input conditions and/or in a multi--period optimization framework. 
In manufacturing, for instance, a producer may activate contracts with customers for a periodic (e.g.\ monthly) supply of its products within an expanding production plan. 
The production capacity is increased in each time period and the goal is to decide which contracts (and when) should be set in order to maximize the overall profits over a given time horizon. 
Here, the increasing capacity values represent the available production resources in each time period while the items are the orders to satisfy with corresponding profits and resource consumption. 
Other applications of an online \IIKP variant in the contexts of renewable resources and of trading operations are provided in \cite{Thtiwe16}.

\subsection{Related literature}

From a general point of view, \cite{HaSh06} introduced incremental versions of maximum flow, bipartite matching, and also knapsack problems. 
The authors in \cite{HaSh06} discuss the complexity of these problems and show how the incremental version even of a polynomial time solvable problem, like the max flow problem, turns out to be NP--hard. General techniques to adapt the algorithms for the considered optimization problems to their respective incremental versions are discussed. Also, a general purpose approximation algorithm is introduced. 

In \cite{BiSeYe13}, it is shown that even  \IIKP is strongly NP--hard. 
A constant factor algorithm is also provided under mild restrictions on the growth rate of the knapsack capacity.
In addition, a PTAS is derived for the special variant \IIKP under the assumption that $T$ is in $O(\sqrt{\log n})$, where $n$ is the number of items.  
In a recent contribution \cite{FaMa17}, an improved PTAS for \IIKP is lined out. The PTAS in \cite{FaMa17} polynomially depends on both $n$ and $T$ without any further restriction on the input data.  Both algorithms introduced in \cite{BiSeYe13} and \cite{FaMa17} cannot be adapted to tackle arbitrary $\mbox{IKP}$ instances. The PTAS derived in \cite{BiSeYe13} could be extended (for fixed $T$) to $\mbox{IKP}$ instances with monotonically non-increasing time factors. 
Likewise, the PTAS proposed in \cite{FaMa17} can be extended  to $\mbox{IKP}$ instances with monotonically increasing time multipliers.

A number of problems closely related to $\mbox{IKP}$ were tackled in the literature. In \cite{faa81, duwa85}, a multiperiod variant of the knapsack problem is considered. 
The capacity increases over periods as in $\mbox{IKP}$ but each item becomes available for packing only at a certain time period.
Thus, the set of items to choose from increases over time.
In \cite{faa81}, a branch and bound algorithm is proposed. In \cite{duwa85}, an efficient algorithm for solving the linear programming relaxation of the problem and a procedure for reducing the number of variables are presented.

In \cite{Thtiwe16}, an online knapsack problem with incremental capacity is considered. In each period a set of items is revealed without knowledge of future items and the goal is to maximize the overall value of the accepted items. The authors in \cite{Thtiwe16} introduce deterministic and randomized algorithms for this online knapsack problem and provide an analysis of their performance.

\subsection{Our contribution}

In this work, we provide approximation results for $\mbox{IKP}$ and its restricted variants. 
We employ Linear Programming (LP) to analyze the worst case performance of different algorithms for deriving approximation bounds. 
Interestingly, the proposed LP--based analysis can be seen as a valid alternative to formal proof systems based on analytical derivation. 
Indeed, recently a growing attention has been centered on the use of LP modeling for the derivation of formal proofs (see \cite{CW16,ACCEHMW16}) and we also show here a successful application of this technique.

In our contribution, in Section~\ref{sec:ApproxIKP} we first generalize for $\mbox{IKP}$ the performance analysis of the generic algorithmic approach derived in \cite{HaSh06} and originally applied to \IIKP. Moreover, we show the tightness of some approximation ratios provided by the approach, including approximation bounds derived in \cite{HaSh06}. 
Then, we devise a PTAS in Section~\ref{IKP_PTAS} when the number of time periods $T$ is a constant. 
While this is a stronger assumption than the one made for the PTAS for \IIKP in \cite{BiSeYe13} 
 and no assumption on input $T$ is needed in \cite{FaMa17}, our algorithm 
is, to the authors' knowledge, the first PTAS that applies to arbitrary $\mbox{IKP}$ instances. Also, the proposed approach is much simpler and does not require a huge number of complicated LP models as in  \cite{BiSeYe13, FaMa17}.

In Section~\ref{sec:RestrIKP} we introduce the natural assumption that each item can be packed at any point in time, also in the first period. 
For this reasonable variant of $\mbox{IKP}$ we discuss two approximation algorithms suited for any $T$. 
Finally, in Section~\ref{sec:RestrIKP23} we focus on an $\mbox{IKP}$ variant with $T=2$ and show an algorithm with 
an approximation ratio bounded by $\frac{1}{2} + \frac{\sqrt{2}}{4} = 0.853\ldots$ and by $\frac{6}{7}=0.857\ldots$ for \IIKP. A preliminary conference version containing some of the results of this paper derived for \IIKP is presented in \cite{DEPFSC17}.

\section{Notation and problem formulation}
\label{sec:themodelIKP}

In $\mbox{IKP}$  a set of $n$ items is given together with a knapsack with 
capacity values $c_t$ for each time period $t=1,\ldots,T$ 
and increasing over time, i.e.\ $c_{t-1} \leq c_t$ for $t=2,\ldots,T$.
Each item $i$ has a profit $p_{i} > 0$ and a weight $w_{i} > 0$. If an item is placed in the knapsack at time $t$, it cannot be removed at a later time. The contribution of an item $i$ at time $t$ is equal to $\Delta_t \, p_{i}$, where  $\Delta_t > 0$ denotes the time multiplier of period $t$. The problem calls for maximizing the total profit of the selected items without exceeding the knapsack capacity summing up the profits attained for each time period.
In order to derive an ILP-formulation, we associate with each item $i$ a $0-1$ variable $x_{it}$ such that $x_{it}=1$ iff item $i$ is contained in the knapsack in period $t$.
$\mbox{IKP}$ can be formulated by the following ILP model: 
\begin{align}
   (\mbox{IKP}) \qquad \text{maximize}\quad & \sum\limits_{t=1}^T \sum\limits_{i=1}^{n} \Delta_t p_{i}x_{it}\label{eq:ObjIKP}\\
	\text{subject to}\quad
	& \sum\limits_{i=1}^{n} w_{i}x_{it} \leq c_t \quad t= 1,\dots,T; \label{eq:capIKP}\\
	& x_{i(t-1)} \leq x_{it} \qquad i= 1,\dots,n, \quad t= 2,\dots,T; \label{eq:PredIKP}\\
	& x_{it} \in\{0,1\} \qquad i= 1,\dots,n, \quad t= 1,\dots,T. \label{eq:varDefIKP}
\end{align}
The cost function (\ref{eq:ObjIKP}) maximizes the sum of the profits over the time horizon. Constraints (\ref{eq:capIKP}) guarantee that the items weights sum does not exceed capacity $c_t$ in each period $t$. Constraints (\ref{eq:PredIKP}) ensure that an item chosen at time $t$ cannot be removed afterwards. Constraints (\ref{eq:varDefIKP}) indicate that all variables are binary. 
In \IIKP, we have $\Delta_t =1$ for each $t=1,\dots,T$. \\

\smallskip

In the following, we denote by $x^*$ (resp.\ $z^*$) the optimal solution (resp.\ solution value) of $\mbox{IKP}$. The contribution of time period $t$ to $z^*$ is denoted by $z^*_t =  \sum_{i=1}^{n} \Delta_t p_{i} x^*_{it}$. We also denote by $z^{Y}$ the solution value yielded by a generic algorithm $Y$ and by $pc_j = \sum_{t=1}^T \Delta_t p_{j} x_{jt}$ the {\em profit contribution} of an item $j$ of a given feasible solution $x$. 
For each period $t$ and the related capacity value $c_t$, we define the corresponding standard knapsack problem as $KP_t$. This means that in $KP_t$ we consider only one of the $T$ constraints (\ref{eq:capIKP}). The optimal solution value and related item set of each $KP_t$ is denoted by $z_t$ and \KPt respectively. Finally, we define for a generic item set $S$ the canonical weight sum $w(S) := \sum_{j \in S} w_j$ and profit sum $p(S) :=\sum_{j \in S} p_j$.
We say that an approximation algorithm $Y$ has {\em approximation ratio} $r$, 
for $r\in [0,1)$, if $z^Y \geq r\cdot z^*$ for every instance of $\mbox{IKP}$.

\subsection{IKP Linear Relaxation}
The linear relaxation of $\mbox{IKP}$, where fractions of items can be packed, 
i.e.\  where constraints (\ref{eq:varDefIKP}) are replaced by the inclusion in the interval $[0,1]$, can be easily computed. 
In fact, it suffices to order the items by non--increasing efficiencies $\frac{p_i}{w_i}$ in $O(n \log n)$ and to fill the capacity of the knapsack in each period according to this ordering in $O(n)$. 
Thus, the execution time for solving the linear relaxation of $\mbox{IKP}$, 
hereafter denoted by $\mbox{IKP}_{LP}$, 
can be bounded by $O(n \log n+T)$. 

An alternative approach would consider each time period $t$ separately and solve 
the LP-relaxation for each knapsack problem $KP_t$.
Using the linear-time median algorithm to find the split item,
this can be done in $O(n)$ time for every $t=1,\ldots, T$, thus requiring $O(n\cdot T)$ time in total.
Clearly, both approaches deliver exactly the same solution structure.
For small or even constant values of $T$ the latter approach will dominate the former.

Taking a closer look at the $T$ iterations of this approach one can exploit the fact that the split item (and thus the LP-relaxation) is always computed on the same item set, 
only the capacities change.
Avoiding repetitions 
one can even get a total running time of $O(n \log T)$.
Since $T$ can be expected to be of moderate size, this is close to a linear time algorithm.
The technical details of this approach can be briefly described as follows:

The classical computation of the split item in linear time requires a sequence of $\log n$
iterations (see~\cite[ch~3.1]{KePfPi04}). 
In iteration $i$ for $i=1,\ldots, \log n$, there are $n/2^{i-1}$ items considered, 
their median w.r.t.\ the efficiencies is computed in $O(n/2^{i-1})$ time,
and the considered items are partitioned into two subsets with efficiencies lower and higher
than the median, respectively.
Then, it is determined (in constant time) which of the two subsets contains the split item
and the next iteration is performed on the selected subset with $n/2^i$ items.

Applying this procedure several times for different capacity values on the same set of items
does not change the positions of medians and the resulting subset structure.
Of course, the selections of subsets in each iteration may change
for different capacities.
Nevertheless, once a certain subset was considered, its median computed and the partitioning into two
halves performed, this process does not have to be repeated for the same subset again
for a different capacity value.
Thus, each subset of size $n/2^i$ has to be processed (in linear time) only once.

We split the analysis of the total effort required for all $T$ capacity values into two phases. We first consider the iterations $i=1, \ldots, \log T$. 
In each such iteration a subset of size $n/2^{i-1}$ is considered.
However, over all $T$ iterations each item can occur only once in a subset of size 
$n/2^{i-1}$.
Thus, the total effort for the $T$ executions of an iteration indexed by $i$ is bounded by $O(n)$.
The time for this first phase is thus $O(n \log T)$.
The running time for the second phase can be simply summed up by
$$ T \cdot \sum_{i=\log T+1}^{\log n} \frac{n}{2^{i-1}} 
\leq T \cdot 2\: \frac{n}{2^{\log T}} = 2n\,. $$
We summarize our considerations in the following statement.
\begin{theorem}
\label{th:lprelax}
The LP-relaxation of $\mbox{IKP}$ can be computed in \\
$O(\min\{n \log T, n \log n+T\})$.
\end{theorem}

\section{Approximating IKP}
\label{sec:ApproxIKP}

\subsection{Approximation ratios of a general purpose algorithm}
\label{sec:HarmNumRatios}
In \cite{HaSh06}, a general framework for deriving approximation algorithms is provided. 
Following the scheme in \cite{HaSh06} which was originally applied to \IIKP, we consider the following algorithm $A$. 
The algorithm employs an $\eps$--approximation scheme to obtain a feasible solution for each knapsack problem $KP_t$ with solution value denoted by $z_{t}^{A}$ for $t=1, \ldots, T$.
Each such solution is also a feasible solution for $\mbox{IKP}$ where 
the items constituting $z_t^A$ are present in all successive time periods (and no other items). 
The algorithm chooses as a solution value $z^A$ the maximum among all these candidates, i.e. 
\begin{equation}\label{eq:z^a}
z^A = \max_{t=1,\dots,T}\left\{\sum_{\tau=t}^T \Delta_\tau \, z_{t}^{A}\right\}.
\end{equation}
The obvious extension of $A$, which solves each knapsack problem $KP_t$ to optimality with value $z_t$ (instead of $z_{t}^{A}$) and then proceeds as in (\ref{eq:z^a}), 
will be denoted by $A^*$ with solution values $z^{A^*}$.
Clearly, $A^*$ is not polynomial.

Let us define quantity
$\Theta = {{‎‎\sum}}_{t=1}^T \frac{\Delta_t}{\sum_{\tau=t}^T \Delta_\tau}$
, with $\Theta \leq T$. 
The following generalization of Theorem 3 in \cite{HaSh06} holds.
\begin{theorem}
\label{ShHaGEN}
Algorithm $A$ is an approximation algorithm for $\mbox{IKP}$ with approximation ratio at least $\frac{1-\eps}{\Theta} \geq \frac{1-\eps}{T}$. 
\end{theorem}
\begin{proof}
For each $t=1,\dots,T$ we have $\frac{\Delta_t z_{t}^{A}}{1 -\eps} \geq z_t^{*}$ and also $z_{t}^{A} \leq \frac{z^A}{\sum\limits_{\tau=t}^T \Delta_\tau}$ from \eqref{eq:z^a}. Then, the following series of inequalities holds:
$$  z^* = \sum\limits_{t=1}^T z_t^{*} \leq  \sum\limits_{t=1}^T   \frac{\Delta_t z_{t}^{A}}{(1 -\eps)} \leq   \sum\limits_{t=1}^T \frac{\Delta_t z^{A}}{(1 -\eps)\sum\limits_{\tau=t}^T \Delta_\tau} = \frac{\Theta \, z^{A}}{1-\eps} \implies  \frac{z^{A}}{z^*} \geq \frac{1-\eps}{\Theta}$$
\end{proof}

For \IIKP considered  in \cite{HaSh06}, we have $\Theta=\mathcal{H}_T$ with $\mathcal{H}_T\approx \ln T$ being the $T$-th harmonic number.
At the present state of the art, the tightness of this bounds is an open question even for \IIKP. 
Considering the variant $A^*$, where each $KP_t$ is solved to optimality (i.e.\ algorithm $A$ for $\eps = 0$), we can state the following result. 
\begin{theorem}
\label{GenTight}
For any value of $T$ and any choice of time multipliers $\Delta_t$, 
algorithm $A^*$ has a tight approximation ratio $\frac{1}{\Theta}$ for $\mbox{IKP}$.
\end{theorem}
\begin{proof}
We can evaluate the performance of algorithm $A^*$ by an alternative analysis based on a Linear Programming model. 
More precisely, we consider an LP formulation with non-negative variables $h^A$ and $h_t$ associated with $z^{A^*}$ and $z_t$ respectively 
and a positive parameter $OPT > 0$ associated with $z^*$. 
The corresponding LP model for evaluating the worst case performance of algorithm $A^*$ is as follows:
\begin{align}
   \text{minimize}\quad & h^A \label{eq:ObjAlgoA}\\
	\text{subject to}\quad
	& h^A  - \sum\limits_{\tau=t}^T \Delta_\tau \, h_t  \geq 0 \qquad t= 1,\dots,T; \label{eq:MaxValueA}\\
	& \sum\limits_{t=1}^T \Delta_t h_t \geq OPT \label{eq:UBIKP}\\
	& h^A \geq 0 \label{eq:hAdef}\\
	& h_t \geq 0 \qquad t= 1,\dots,T. \label{eq:hTdef}
\end{align}
The value of the objective function (\ref{eq:ObjAlgoA}) provides a lower bound on the worst case performance of algorithm~$A^*$.
Constraints (\ref{eq:MaxValueA}) guarantee that the contribution of each optimal knapsack solution $KP_t$ as a solution of $\mbox{IKP}$ will be taken into account according to (\ref{eq:z^a}). 
Constraint (\ref{eq:UBIKP}) indicates that the sum of the knapsack solutions $z_t$, multiplied by the associated $\Delta_t$, constitute a trivial upper bound on $z^*$. 
Constraints (\ref{eq:hAdef}), (\ref{eq:hTdef}) indicate that variables are non-negative. Solving model (\ref{eq:ObjAlgoA})--(\ref{eq:hTdef}) to optimality provides lower bounds on the performance ratio of algorithm $A^*$ equal to $\frac{h^A}{OPT}$ for any $T$.
We will derive the optimal solutions of model \eqref{eq:ObjAlgoA}--\eqref{eq:hTdef} by considering the related dual problem. Let us introduce $T +1$ dual variables $\lambda_i ~ (i=1,\dots,T+1)$ associated with constraints \eqref{eq:MaxValueA}--\eqref{eq:UBIKP} respectively. The dual formulation of model \eqref{eq:ObjAlgoA}--\eqref{eq:hTdef} is as follows: 
\begin{align}
   \text{maximize}\quad & OPT \, \lambda_{(T+1)} \label{eq:ObjDUALAlgoA}\\
	\text{subject to}\quad
	& \sum\limits_{t=1}^T \Delta_t\leq 1 \label{eq:hAA}\\
	& -\sum\limits_{\tau=t}^T \Delta_\tau \, \lambda_t + \Delta_t \lambda_{(T+1)} \leq 0 \qquad t= 1,\dots,T; \label{eq:htt}\\
	& \lambda_t \geq 0 \qquad t= 1,\dots,T+1. \label{eq:lambdadef}
\end{align}
Constraints \eqref{eq:hAA}--\eqref{eq:htt} correspond to primal variables $h^A$, $h_t  ~ (t= 1,\dots,T)$ respectively.
Feasible solutions of primal and dual models are:
\begin{align}
& h^{A} = \frac{OPT}{\Theta}; \, h_t = \frac{OPT}{\sum\limits_{\tau=t}^T\Delta_\tau \, \Theta}  \qquad t= 1,\dots,T; \label{PrimalGen}\\
&\lambda_{(T + 1)} = \frac{1}{\Theta}; \, \lambda_t = \frac{\Delta_t}{\sum\limits_{\tau=t}^T\Delta_\tau \, \Theta}  \qquad t= 1,\dots,T. \label{DualGen}
\end{align}
It is easy to verify that such solutions are feasible and satisfy all corresponding constraints to equality for any distribution of time multipliers and for any value of $T$. Since $h^{A} = OPT \, \lambda_{(T+1)}$, by strong duality both primal and dual solutions are also optimal. Hence, the bound $\frac{h^A}{OPT}$ is equal to $\frac{1}{\Theta}$ for any $T$. 

It was shown in Theorem~\ref{ShHaGEN} (for $\eps=0$ ) that the approximation ratio of $A^*$ is at least $\frac{1}{\Theta}$.
Now we will derive instances where this bound can be reached.
Therefore, we construct instances where the optimal solution values of $KP_t$ in each period $t$ are equal to the values of $h_t$ for $t=1,\dots, T$ in \eqref{PrimalGen} and where the optimal solution value for $\mbox{IKP}$ is equal to the weighted sum of all these solutions, namely $z^* = \sum\limits_{t=1}^T \Delta_t h_t$. Such target instances can be generated by the following procedure:
\begin{enumerate}  
\item We first represent quantity $\Theta$ as a fraction, 
i.e.\ $\Theta = \frac{a}{b}$ where b is the smallest common multiple of the denominators of fractions $\frac{\Delta_1}{\Delta_1 + \dots + \Delta_T} +  \dots + \frac{\Delta_{T}}{\Delta_T}$. Then, we set $OPT = a$ so as to get integer values $h_t$.
\item We generate an $\mbox{IKP}$ instance as follows:
$$n = \frac{b}{\Delta_T},\ p_j = w_j = 1\ (j=1,\dots, n),\ c_t = h_t \, (t=1,\dots, T)$$
The optimal solution of each $KP_t$ will pack items until the corresponding capacity $c_t$ is filled and thus will yield a solution value equal to $h_t$. 
The number of items is $\frac{b}{\Delta_T}$ because the capacity in the last period $T$ is $c_T = h_T = \frac{OPT}{\Delta_T \Theta}= \frac{a}{\Delta_T \frac{a}{b}} = \frac{b}{\Delta_T}$. At the same time, the optimal solution for $\mbox{IKP}$ can be obtained by progressively packing all items over time periods while filling the capacities $c_t$, 
hence $z^* = \sum\limits_{t=1}^T \Delta_t c_t = \sum\limits_{t=1}^T \Delta_t h_t$. 
\end{enumerate}  
\end{proof}
As an example of the outlined procedure, consider the \IIKP variant with $T=3$ for which $\mathcal{H}_T = \frac{11}{6}$. Correspondingly, the following instance with $n=6$, $p_j = w_j = 1$ ($j=1,\dots, 6$), $c_1 = 2$, $c_2 = 3$, $c_3 = 6$ is generated.
The optimal solution is given by greedily packing as many items as possible in each time period and thus filling the corresponding capacities ($z^* = 11$). The optimal solutions values of the KPs are equal to  2, 3, 6 respectively. 
Hence we have $z^{A^*} =  \max \{3\cdot2, 2\cdot3, 6 \} = 6$ which proves the tightness of the approximation bound $\frac{1}{\mathcal{H}_3}  = \frac{6}{11}$.


We remark that the bound tightness cannot be straightforwardly generalized to algorithm $A$, where an $\eps$--approximation scheme is run for solving each KP. 
We could get the ratio in Theorem \ref{ShHaGEN} by solving model (\ref{eq:ObjAlgoA})--(\ref{eq:hTdef}) where the term $\sum\limits_{t=1}^T \Delta_t h_t$  in constraint (\ref{eq:UBIKP}) is divided by ($1 - \eps$) and the optimal values of primal variables are multiplied by $(1 - \eps)$. However, the generation of tight instances is strictly related to the choice of the approximation algorithm for KPs.

\subsection{A PTAS when T is a constant} 
\label{IKP_PTAS}



In the following it will be convenient to consider a {\em residual problem} $\mbox{IKP}^R$ with optimal solution value $z_R^*$.
It is derived from any given instance of $\mbox{IKP}$ by defining a subset of items, 
a residual capacity $c^R_t \leq c_t$ for every time period $t$ with $c^R_{t-1} \leq c^R_t$,
and an {\em earliest insertion time} $t_j$, which means that item $j$ cannot be put into the knapsack in time periods before $t_j$.
Formally, the conditions $x_{j \tau}=0$ for $\tau < t_j$ are added in the ILP model.
Obviously, every $\mbox{IKP}$ instance is equivalent to a restricted instance $\mbox{IKP}^R$ with $t_j := \min \{t \mid c_t \geq w_j, t=1,\ldots, T\}$.
Clearly, the solution value of the LP-relaxation of this $\mbox{IKP}^R$ can be smaller than the LP-relaxation of $\mbox{IKP}$.
Concerning the structure of the LP-relaxation of $\mbox{IKP}^R$ we can observe the following.
\begin{lem}\label{th:lpres}
The solution of the LP-relaxation of $\mbox{IKP}^R$ contains at most $T$ fractional items.
\end{lem}
\begin{proof}
We claim that considering the time periods from $t=1,\ldots, T$ at most one additional fractional item can appear in any period $t$ which implies the above statement.

Assume that in some period $t$ there are two new fractional items $j_1$ and $j_2$ in the solution (which were not packed at all in periods before $t$) with efficiencies  $p_{j_1}/w_{j_1} \geq p_{j_2}/w_{j_2}$. 
Similar to the standard KP (see~\cite[Theorem 2.2.1]{KePfPi04})
one can invoke an exchange argument and shift weight from $j_2$ to $j_1$ by setting 
$d:=\min\{w_{j_2}x_{j_2 t}, w_{j_1}(1-x_{j_1 t}) \}$
and defining a new solution with 
$x_{j_1 t}' := x_{j_1 t} + d/w_{j_1}$
and
$x_{j_2 t}' := x_{j_2 t} - d/w_{j_2}$.
Obviously, the total capacity consumed by items $j_1$ and $j_2$ does not change in period $t$ and neither in any later periods up to $T$.
Hence, the new solution is feasible and gives a higher profit contribution in all periods $t, t+1,\ldots, T$.
Moreover, in the new solution at most one of $j_1$ and $j_2$ has a fractional value in period $t$.
\end{proof}
Note that a fractional item introduced in some period $t$ may well be increased during a later period and possibly reach $1$.
Moreover, different from the LP-relaxation for $\mbox{IKP}$, it may be the case that the capacity is not fully utilized in some periods.
This can occur e.g.\ in some period $t$, if a highly efficient item with large weight has an earliest insertion time larger than $t$.
In this case it can make sense to leave some empty space in period $t$ and thus allow this item to be fully packed when it becomes available.

Taking a closer look at the solution structure of the LP-relaxation of $\mbox{IKP}^R$ we can also compute it by a combinatorial algorithm with polynomial running time. However, the details of this algorithm are beyond the scope of this paper.

\medskip
Computing the LP-relaxation of any restricted instance $\mbox{IKP}^R$ and rounding down all fractional values gives a feasible solution with value $z^{\prime}_R$.
It follows from Lemma~\ref{th:lpres} that the maximum loss from rounding down is bounded by $T$ times the maximum profit contribution of a single item.
Thus we have:
\begin{equation}\label{eq:zAprimeR}
z_R^* \leq z^{\prime}_R + T \cdot \max_{j=1,\ldots,n} \sum_{\tau=t_j}^T \Delta_\tau p_j
\end{equation}

For a residual problem $\mbox{IKP}^R$ we can apply the following approximation algorithm $A^{\prime}$, which is a variant of algorithm $A$ described in Section~\ref{sec:HarmNumRatios}.
We run an FPTAS for each time period $t$ yielding $\eps$-approximations $z_{tR}^{A^\prime}$.
Then we also consider an alternative solution 
derived by computing the optimal solution of the LP-relaxation of $\mbox{IKP}^R$ and rounding down all fractional variables to $0$ as described above
reaching a solution value $z^{\prime}_R$.
Finally, we take the maximum between these $T+1$ candidates obtaining a solution value $z_R^{A^{\prime}}$, namely 
\begin{equation}\label{eq:z^B}
z_R^{A^{\prime}} = \max\left\{ z^{\prime}_R,\: 
\max_{t=1,\ldots,T} 
\left\{ \sum_{\tau=t}^T \Delta_\tau z_{tR}^{A^\prime}\right\}
\right\}.
\end{equation}
With Theorem~\ref{ShHaGEN} the overall approximation ratio of algorithm $A^{\prime}$ for $\mbox{IKP}^R$ can be bounded by $\rho = \frac{1-\eps}{T}$.
Note that the introduction of earliest insertion times for items in $\mbox{IKP}^R$ does not interfere with the analysis in Theorem~\ref{ShHaGEN}.

\smallskip
Computing $z^{\prime}_R$ can be done in 
polynomial time $poly(n,T)$.
The $T$ executions of a FPTAS for KP can be bounded by
$O(T(n \log(\frac{1}{\eps}) + (\frac{1}{\eps})^3 \log^2(\frac{1}{\eps})))$,
see~\cite{KelPfe04}.
 
\bigskip
Similarly to the line of reasoning for deriving PTAS's for KP 
(see, e.g., \cite{Sahni75, CaKePfPi00}), 
we propose an approximation scheme for $\mbox{IKP}$.
It is based on guessing (by going through all possible choices) the set $S_k$ of $k$ items and the associated starting periods $st_j$ for each item $j \in S_k$ which give the largest profit contributions in an optimal solution. 
For convenience, the minimum profit contribution of an item in $S_k$ is denoted by
$pcm = \min_{j\in S_k} \{\sum_{\tau=st_j}^T \Delta_\tau p_j\}$.

For each such choice of $S_k$ and starting periods, if it is feasible, one can define the following residual instance $\mbox{IKP}^R$.
It consists of the residual capacities $c^R_t = c_t - \sum_{j \in S_k, st_j \leq t} w_j$ 
in each time period $t$.
To preserve the incremental capacity structure of $\mbox{IKP}$
we set $c^R_t := \min\{c^R_t, c^R_{t+1}\}$ for $t=T-1, T-2, \ldots, 1$,
which is without loss of generality because of (\ref{eq:PredIKP}).
The item set of $\mbox{IKP}^R$ is restricted to items not in $S_k$ with profit contribution not exceeding $pcm$, which can be enforced by setting the earliest insertion time $t_j$ for all items $j \not \in S_k$ as  
$$t_j = \min \left\{t \mid c^R_t \geq w_j, \sum_{\tau=t}^T \Delta_\tau p_j \leq pcm,\, t=1,\ldots, T \right\}.$$
If the minimum is taken over the empty set, we remove $j$ from $\mbox{IKP}^R$.
We can now state our approximation scheme, denoted by algorithm \textit{Approx}.

\begin{algorithm}[H]
\begin{algorithmic}[1]
\STATEx \textbf{Input:} $\mbox{IKP}$ instance, $\eps$. 
\STATE \label{alg:1} 
Set $k := \min \left \{ n, \Bigl\lceil \frac{T}{\eps} \Bigr\rceil \right \}$.
\STATE \label{alg:2} 
Generate all subsets of $\{1,\ldots,n\}$ consisting of less than $k$ items.\newline
For each such subset generate all possible starting periods of items and check the feasibility for $\mbox{IKP}$.\newline
Let $z^\eps$ be the best objective value over all generated feasible solutions.
\STATE \label{alg:3} 
For each subset $S_k \subseteq \{1,\ldots,n\}$ with $|S_k|=k$, generate all possible starting periods of items and check the feasibility for $\mbox{IKP}$.
\STATE \label{alg:4} 
For every resulting feasible configuration, 
compute the overall profit contribution $PC(k)=\sum_{j\in S_k} pc_j$ and 
the minimum contribution $pcm$.
Determine the corresponding residual $\mbox{IKP}$ instance $\mbox{IKP}^R$.\newline
Apply algorithm $A^{\prime}$ to $\mbox{IKP}^R$ yielding a solution value $z_R^{A^{\prime}}$.\newline
Update $z^\eps := \max\{z^\eps, z_R^{A^{\prime}} + PC(k)\}$
\end{algorithmic}
\caption{\textbf{Algorithm} \textit{Approx}}
\end{algorithm}

The following proposition characterizes the running time of \textit{Approx}.
\begin{prop}
\label{ptas_complex}
The running time complexity of \textit{Approx} is polynomial in the size $n$ of the input.
\end{prop}

\begin{proof}
The running time mainly depends on the number of configurations generated in Step~\ref{alg:3}.
There are $O(n^k)$ subsets $S_k$ and for each item $T$ possible starting periods, i.e.\ $O(T^k)$ possible choices for each $S_k$. 
For each such configuration algorithm \textit{$A^{\prime}$} is performed in Step~\ref{alg:4}.
This dominates the effort spent in Step~\ref{alg:2}.	
Plugging in the definition of $k$ the overall complexity is:
\begin{equation}
O\left(n^{\lceil \frac{T}{\eps} \rceil} T^{\lceil \frac{T}{\eps} \rceil}
\cdot
\left( poly(n,T) +
T\left(n \log(\frac{1}{\eps}) + (\frac{1}{\eps})^3 \log^2(\frac{1}{\eps})\right)\right)\right)
\end{equation}
\end{proof}

The following theorem describes the approximation ratio reached by \textit{Approx}.
\begin{theorem}
Algorithm \textit{Approx} is a PTAS for $\mbox{IKP}$ when $T$ is a constant.
\end{theorem}

\begin{proof}
We will show in this proof that the algorithm is an $\eps$--approximation scheme. 
Then, it follows from the time complexity stated in Proposition \ref{ptas_complex} that \textit{Approx} is a PTAS for $\mbox{IKP}$ when $T$ is a constant.

If $k = n$, namely $ n \leq \Bigl\lceil \frac{T}{\eps} \Bigr\rceil$, 
Steps~\ref{alg:2} or \ref{alg:3} of \textit{Approx} would even compute an optimal solution since all subsets of the $n$ items are enumerated, and therefore the statement is trivial.
Thus, it remains to consider the case where $k=\Bigl\lceil \frac{T}{\eps} \Bigr\rceil$.

We consider two cases depending on a parameter $f \in (0,1)$ and analyze the iteration of Step~\ref{alg:3} where $S_k$ and the associated starting periods $st_j$ yield the largest profit contribution $PC(k)$ of any subset of $k$ items in the optimal solution.
In the corresponding execution of Step~\ref{alg:4} we denote by $z_R^*$ and by $p_{max}^R$ the optimal solution value and the maximum profit of the items in the respective residual instance $R$. 

\paragraph{Case 1: $PC(k) \geq f\cdot z^*$} \mbox{} \\
Since the $k$ items with maximal profit contribution in the optimal solution
are considered, we have $z^* = z_R^* + PC(k)$. 
Using the approximation ratio $\rho$ of $A^{\prime}$, the following series of inequalities holds: 
\begin{equation}
\begin{aligned}
\label{case1}
&z^\eps \geq PC(k) + z_R^{A^{\prime}} \geq PC(k) + \rho\, z_R^* \\
& = PC(k) + \rho (z^* - PC(k)) = (1-\rho)PC(k) + \rho z^*\\
& \geq (1-\rho)fz^* + \rho z^* = ((1-\rho)f + \rho)\, z^* 
\end{aligned}
\end{equation}

\paragraph{Case 2: $PC(k) < f\cdot z^*$} \mbox{} \\
Since the profit contribution of any item in $\mbox{IKP}^R$ is bounded by $pcm$,
we also for any item $j$ in $\mbox{IKP}^R$ 
\begin{equation}\label{eq:pmaxR}
pc_j \leq pcm \leq \frac{1}{k} PC(k) <  \frac{1}{k} f\,z^*. 
\end{equation}
Then, the following series of inequality holds: 
\begin{eqnarray}
\label{eq:case2}
z^* &=& PC(k) + z_R^*\\
&\leq& PC(k) + z_R^{A^{\prime}} + T \cdot \max_{j \not\in S_k}\, pc_j \\
&\leq& PC(k) + z_R^{A^{\prime}} + T \cdot \frac{1}{k} f\,z^*\\
&\leq& z^\eps + T\frac{1}{k}f\,z^* \label{eq:lastline}
\end{eqnarray}
The first inequality comes from (\ref{eq:zAprimeR}),
the second inequality from (\ref{eq:pmaxR}).
Now (\ref{eq:case2})-(\ref{eq:lastline}) is equivalent to
\begin{equation}\label{eq:FinalCase2}
z^\eps \geq (1 - T\frac{1}{k}f)\,z^*. 
\end{equation}
Given $\eps$ and $\rho$, we now set $f$ as follows:
\begin{equation}
f := 1 - \frac{\eps}{1-\rho} = \frac{1-\rho - \eps}{1-\rho} \label{eq:setf}
\end{equation}
We easily note that $f < 1$ and we also have $f > 0$ since 
$\rho = \frac{1-\eps}{T}$ and thus $1-\eps > \rho$. 
For \textit{Case 1}, plugging the value of $f$ in (\ref{case1}) yields
\begin{equation}
\begin{aligned}
\label{eq:checkcase1}
&z^\eps \geq ((1-\rho)f + \rho)\, z^* \\
&= ((1-\rho - \eps) + \rho)\, z^* = (1 - \eps)\,z^*.
\end{aligned}
\end{equation}
For \textit{Case 2}, we plug in $k=\Bigl\lceil \frac{T}{\eps} \Bigr\rceil$ into (\ref{eq:FinalCase2}) 
and get
\begin{equation}\label{eq:checkcase2}
z^\eps \geq (1 - \eps f)\, z^* \geq (1 - \eps)\, z^*.
\end{equation}
\end{proof}

We remark that in \cite{BiSeYe13} a PTAS is introduced under the more general assumption that $T$ is in $O(\sqrt{\log n})$,
but only for \IIKP.
For constant $T$, it can be extended to decreasing time multipliers $\Delta_t$.
Furthermore, in \cite{FaMa17} a PTAS is given for \IIKP for arbitrary $T$
and for $\mbox{IKP}$ with increasing time multipliers.
All these PTASs rely on solving a huge number
of non--trivial LP models. 
Our approach works for general multipliers $\Delta_t$ and does not require the solution of LP models.

\section{Approximation algorithms for a weight constrained IKP variant}
\label{sec:RestrIKP}
There are two essential decisions in $\mbox{IKP}$, namely whether to select an item at all,
and -- if yes -- when to select it.
Thus, it seems natural not to restrict the latter decision to a subset of time periods
and allow each item to be selected at any point in time, also in period $1$.
Therefore, in the reminder of the paper we will consider $\mbox{IKP}$ under the mild and natural assumption 
that each item can be packed in the first period, i.e.\ $w_i \leq c_1$, $i=1,\dots,n$. 
We refer to this weight constrained variant of the problem as $\mbox{IKP}^\prime$ and to the related version with unit time multipliers $\Delta_t=1$ for all $t$ as $\mbox{IIKP}^\prime$.

\medskip
Let us denote by $s_t$ the split items in the linear relaxation of each $KP_t$ for $t=1,\dots,T$, which also correspond to the fractional items in the optimal solution of $\mbox{IKP}^\prime_{LP}$. 
We state the following algorithm $H_1$, which is independent from the given multipliers $\Delta_t$.

\begin{algorithm}[H]
\begin{algorithmic}[1]
\STATEx \textbf{Input:} $\mbox{IKP}^\prime$ instance.
\STATE Sort items by decreasing $\frac{p_i}{w_i}$ $(i = 1,\dots,n)$ and solve $\mbox{IKP}^\prime_{LP}$.
\STATE 
Let $\hat{t} = \min\{ t \mid c_t \geq \sum_{j=1}^{s_1} w_j,\, t=1,\ldots, T \}$. \newline
Let $p=\max \left \{\sum_{j=1}^{s_1 - 1}  p_j, p_{s_1} \right \} $.
\STATE Pack the item(s) yielding profit $p$ in time periods $1$ to $\hat{t}-1$.\newline
Pack all items $j=1,\dots, s_1$ in time periods $\hat{t}$ to $T$.
\STATE Add the remaining items of the optimal solution of $\mbox{IKP}^\prime_{LP}$ 
ignoring all fractional values.   
\end{algorithmic}
\caption{\textbf{Algorithm $H_1$}}
\end{algorithm}

The following theorem holds.
\begin{theorem}
Algorithm $H_1$ has a tight $\frac{1}{2}$-approximation ratio for $\mbox{IKP}^\prime$
for any choice of multipliers $\Delta_t$.
\end{theorem}

\begin{proof}
Consider the optimal profit values $z_t$ of $KP_t$ for $t=1,\dots,T$ (without multipliers). 
From the properties of KP and 
since items are ordered by decreasing $\frac{p_i}{w_i}$, we have:
\begin{align}
\max \left \{\sum\limits_{j=1}^{s_t - 1}  p_j,p_{s_t} \right \} \geq \frac{1}{2}\, z_t \qquad t=1,\dots,T \label{eq:genKP}
\end{align}
\begin{align}
\frac{\sum_{j=1}^{s_t - 1} p_j }{\sum_{j=1}^{s_t - 1} w_j} \geq \frac{p_{s_t}}{w_{s_t}} \qquad t=1,\dots,T \label{eq:orditems} 
\end{align}
Algorithm $H_1$ yields a solution with value 
\begin{align}
\label{eq:H_1sol}
z^{H_1} = \sum_{t=1}^{\hat{t} - 1} \Delta_t \cdot \max \left \{\sum\limits_{j=1}^{s_1 - 1} p_j,p_{s_1} \right \} 
+ \sum_{t=\hat{t}}^{T} \Delta_t \cdot \sum\limits_{j=1}^{s_t - 1}  p_j.
\end{align}
Since inequalities $\sum_{j=1}^{s_{t} - 1} w_j \geq \sum_{j=1}^{s_1}w_j > c_1 \geq w_i$ hold for any item $i$ and $t \geq \hat{t}$, from (\ref{eq:orditems}) we get $\sum_{j=1}^{s_t - 1}  p_j > p_{s_t}$ for any $t \geq \hat{t}$ and thus
\begin{align}
\label{eq:afterC1}
\max \left \{\sum\limits_{j=1}^{s_t - 1}  p_j,p_{s_t} \right \} = \sum\limits_{j=1}^{s_t - 1}  p_j \qquad t=\hat{t},\dots,T.
\end{align}
Considering (\ref{eq:genKP}), (\ref{eq:H_1sol}), (\ref{eq:afterC1}) and that 
$\sum_{t=1}^{T} \Delta_t z_t$ is an upper bound on $z^*$, we get
\begin{align}
\label{eq:1over2}
z^{H_1} \geq  \frac{1}{2}\sum_{t=1}^{\hat{t} - 1} \Delta_t z_t + \frac{1}{2}\sum_{t=\hat{t}}^{T} \Delta_t z_t \geq \frac{1}{2}z^*
\end{align}
which shows that algorithm $H_1$ has an approximation ratio of $\frac{1}{2}$.

\medskip
To prove the tightness of the bound, consider the following instance with $c_t = M + t$ for $t=1,\dots,T$ (with integer $M \gg T$) and 3 items with entries
$$p_1 = \frac{M}{2} + \delta, w_1 = \frac{M}{2}; \;  
p_2 = \frac{M}{2} + T + 1, w_2 =  \frac{M}{2} + T + 1; \;
p_3 = \frac{M}{2} - \delta, w_3 =  \frac{M}{2};$$
with $\delta > 0$ being an arbitrary small number. 
Algorithm $H_1$ will select only the second item over all time periods getting a solution with value $\sum_{t=1}^T \Delta_t (\frac{M}{2} +  T + 1)$. 
The optimal $\mbox{IKP}^\prime$ solution will pack items 1 and 3 for all periods
independently from the multipliers $\Delta_t$
reaching a value of $\sum_{t=1}^T \Delta_t \cdot M$. 
Hence, the approximation ratio of the algorithm is
 $$ \frac{\sum_{t=1}^T \Delta_t(\frac{M}{2} + T + 1)}{\sum_{t=1}^T \Delta_t M} = \frac{1}{2} + \frac{T + 1}{M}$$
which reaches a value arbitrarily close to $\frac{1}{2}$ for large values of $M$
and arbitrary multipliers $\Delta_t$.
\end{proof}

Given the approximation ratio of this simple heuristic, it is worth to investigate to what extent computationally more involved iterative approaches could improve upon the worst case performance of algorithm $H_1$. 
We consider an algorithm, denoted as $H_2$, which also considers the periods from $t=1$ to $T$ but computes in each step an {\em optimal} knapsack solution instead of the rounded down LP-relaxation. 

\begin{algorithm}[H]
\renewcommand{\thealgorithm}{}
\begin{algorithmic}[1]
\STATEx \textbf{Input:} $\mbox{IKP}^\prime$ instance.
\STATE Solve $KP_1$ to optimality and pack \KPone 
in all periods $1,\ldots, T$.
\STATE For all periods $t=2, \ldots, T$:\newline
Solve to optimality the knapsack problem induced by $KP_t$ 
with the condition that all items packed in previous periods
remain packed in period $t$.\newline
Pack the additional item set  in period $t$.
\end{algorithmic}
\caption{\textbf{Algorithm $H_2$}}
\end{algorithm}

Contrary to our expectation, we can show the following negative result for algorithm $H_2$.
It turns out quite surprisingly that algorithm $H_2$,
which locally dominates $H_1$ for each period, yields a worse performance ratio for the overall problem for almost any choice of multipliers.
Only for extremely large $\Delta_1$, namely if
$\Delta_1 > \sum_{t=3}^T \Delta_t (t-2) $,
$H_2$ outperforms $H_1$ from a worst case perspective.

\begin{theorem}
Algorithm $H_2$ has a tight approximation ratio of $\frac{\sum_{t=1}^T \Delta_t}{\sum_{t=1}^T \Delta_t\, t} \geq \frac 1 T$ for $\mbox{IKP}^\prime$.
\end{theorem}
Note that for $\mbox{IIKP}^\prime$ 
we get a ratio of $\frac{2}{T+1}$.

\begin{proof}
We first show the lower bound on the performance ratio of $H_2$. 
For every period $t=1,\ldots, T$ we denote by $\pi^*_t$ (resp.\ $\pi_t$) the 
total profit of the items added in period $t$ in the optimal solution of $\mbox{IKP}^\prime$ (resp.\ added by the optimal solution of the residual knapsack problem solved in $H_2$). 
Note that $\pi_1$ is the optimal solution value of $KP_1$. 
For every period $t$ denote by $O_t$ all items selected in period $t$ in the optimal ILP solution and by $I_t$ all items packed by $H_2$
(independently from their starting period and without accounting for the multiplier $\Delta_t$).
The total profit values contributed in period $t$ are thus $p(O_t) = \sum_{j=1}^t \pi^*_j$ and $p(I_t) = \sum_{j=1}^t \pi_j$ respectively.

Now we compare the optimal solution with the outcome generated by $H_2$.
By construction there is $\pi^*_1 \leq \pi_1$. 
We will prove the crucial claim that in each successive time period an additional profit deviation of at most $p_{\max}$ can be accrued, namely:
\begin{equation}\label{eq:indclaim}
\pi^*_{t} - \pi_{t}  \leq  p_{\max} \quad t=2,\ldots, T
\end{equation}
For any period $t \leq T - 1$, we partition the items added in period $t+1$ and contributing to the optimal profit value $\pi^*_{t+1}$ into a set $I^h$, 
which are items also in $I_t$, and a set $I^n$ consisting of items not yet packed by $H_2$.

We now estimate the value of $\pi_{t+1}$ by constructing an auxiliary set $I'$ as follows. Let $I^p:=O_t \setminus I_t$. All items in $O_{t+1} \setminus I_t = I^p \cup I^n$ are clearly available for $\pi_{t+1}$ (possibly among many others). \\
We fill items into $I'$ and stop as soon as $w(I') > c_{t+1} -c_t$. This is done by first considering the items in $I^n$ and then items in $I^p$, both in arbitrary order.

In the unusual case that we do not reach the stopping criterion
we would have $w(I') = w(I^n) + w(I^p) \leq c_{t+1} - c_t$ 
and thus $\pi_{t+1} \geq p(I')$. But then we have $O_{t+1} = I' \cup I_t$ 
and thus even $p(I_{t+1}) \geq p(O_{t+1})$.\\
Otherwise, after removing again the item added most recently into $I'$,
we have a solution at hand which can be packed by $H_2$
since its weight is at most $c_{t+1} -c_t$.
We will prove the following inequality 
\begin{equation}\label{eq:claim}
p(I') \geq \pi^*_{t+1}.
\end{equation}
From (\ref{eq:claim}) our original claim (\ref{eq:indclaim}) follows
because removing an item from $I'$ gives a lower bound for $\pi_{t+1}$ and thus we would have $\pi_{t+1} \geq  p(I') - p_{\max} \geq \pi^*_{t+1} - p_{\max}$.

To prove (\ref{eq:claim}), we first show that $I'$ contains all items of $I^n$.
Assume otherwise that some item $i \in I^n$ is not in $I'$. 
Then we have 
$$c_{t+1} \geq w(O_t) + w(I^n) \geq w(O_t)+ w(I') +w_i >
w(O_t) + c_{t+1} -c_t + w_i\, ,$$
which means that $w(O_t) + w_i < c_t$
and thus item $i$ could have been added already to $O_t$ in contradiction to optimality.

Next we show that $w(I^h)$ cannot be too large.
Denote as $I'':= I' \cap I^p$, i.e.\ the items added to $I'$ which were already 
included in $O_t$. Clearly $I' = I'' \cup I^n$.
If $w(I^h) \geq c_t - w(O_t \setminus I'')$ then we have:
\begin{eqnarray}
w(O_{t+1}) &=& w(O_t \setminus I'') + w(I'') +  w(I^h)+w(I^n) \\
&=& w(O_t \setminus I'') + w(I^h)+ w(I')\\
&>& w(O_t \setminus I'') + w(I^h) + c_{t+1} - c_t\\
&\geq & w(O_t \setminus I'') + c_t - w(O_t \setminus I'') + c_{t+1} - c_t = c_{t+1}
\end{eqnarray}
which is a contradiction. Hence, we must have that $w(I^h) < c_t - w(O_t \setminus I'')$. \\
Therefore, $I^h$ could be used in $O_t$ instead of $I''$. The fact that this was not done by the optimal solution implies that $p(I^h) \leq p(I'')$ which is equivalent to
$\pi^*_{t+1} = p(I^h) + p(I^n) \leq p(I'') + p(I^n) = p(I')$ and (\ref{eq:claim}) is shown.

\smallskip
Thus, considering (\ref{eq:indclaim}), for any period $t \geq 2$ we have 
\begin{equation} 
\label{FinIneq}
\pi^*_1 + \sum_{j=2}^t (\pi^*_j - \pi_j) \leq \pi_1 + (t-1)p_{\max} \implies p(O_t) \leq p(I_t) + (t-1)p_{\max} 
\end{equation}
Summing up (\ref{FinIneq}) over all periods we get
\begin{eqnarray}
z^* &=& \sum_{t=1}^T \Delta_t \, p(O_t) \\
&\leq & \sum_{t=1}^T \Delta_t \left(p(I_t) + (t-1)p_{\max}\right) \\
&=& z^{H_2} + p_{\max} \sum_{t=1}^T \Delta_t(t-1) \\
& \leq& z^{H_2} + \frac{z^{H_2}}{ \sum_{t=1}^T \Delta_t} \cdot \sum_{t=1}^T \Delta_t(t-1) \label{AverZH2}
= z^{H_2} \cdot  \frac{\sum_{t=1}^T \Delta_t\, t}{ \sum_{t=1}^T \Delta_t}
\end{eqnarray}
which shows the stated approximation ratio. 
The last inequality (\ref{AverZH2})  derives from the trivial fact that 
$\pi_1 \geq p_{\max}$ and thus $z^{H_2}  \geq p_{\max} \sum_{t=1}^T \Delta_t$.

The minimum of this bound is reached for multipliers
$\Delta_1 = \ldots = \Delta_{T-1} = \eps$ for small $\eps>0$ and $\Delta_T=1$
where the approximation ratio tends to $\frac 1 T$.
 
\medskip
We can prove the tightness of the bound for arbitrary multipliers $\Delta_t$ by considering the following instance with capacities $c_t = t\cdot M + (T-t)$ and $2\,T$ items with entries 
\begin{eqnarray*}
w_1 = p_1 &=& M + (T-1),\\
w_2 = w_3 = \cdots = w_T &=& M-1,\\
p_2 = p_3 = \cdots = p_T &=& 1,\\
w_{T+1} = w_{T+2} = \cdots = w_{2T} &=& M,\\
p_{T+1} = p_{T+2} = \cdots = p_{2T} &=& M.
\end{eqnarray*}
Note that $c_{t+1} - c_t = M-1$. The solution provided by algorithm $H_2$ starts with $KP_1$ and packs item $1$. For $KP_2$ the residual capacity is $c_2-c_1= M-1$, so only item $2$ with profit $1$ can be packed. This argument continues for each $KP_t$.
Hence, we have as solution value
$$z^{H_2} = \sum_{t=1}^T \Delta_t \cdot (M+(T-1)) + \sum_{t=2}^T \Delta_t (t-1).$$
The optimal $\mbox{IKP}^\prime$ solution packs $t$ copies of items with weight and profit $M$ in each time period $t$ which yields
$$z^* = \sum_{t=1}^T \Delta_t \, t \cdot M. $$
This shows the tightness of the approximation bound for large $M$. 
\end{proof}

\medskip
Mirroring the iterative strategy of algorithm $H_2$, one may consider another natural heuristic approach which starts from the optimal solution of the last knapsack problem $KP_T$ and moves upwards until period $1$ by solving to optimality the knapsack problem induced by $KP_t$ $(t = T -1, \dots, 1)$ with item set restricted to those items packed in period $t+1$. 
Still, it turns out that this strategy cannot improve the approximation ratio of the polynomial time heuristic $H_1$. 
We show this observation
by giving a class of instances where the above algorithm reaches only $z^*/3$.
Consider the following $\mbox{IIKP}^\prime$ instance with 
capacities $c_t = 4 - \gamma$, ($t=1, \dots, T - 2$), $c_{(T - 1)} = 4 + \gamma$, $c_{T} = 5$ (with $\gamma > 0$ being an arbitrary small number) and four items with entries: 
\begin{table}[H]
	\centering
\begin{tabular}{|l|cccc|}
  \hline
$j$ & 1 & 2 & 3 & 4 \\ \hline
$p_j$ & $1 + \gamma$ & $1 + \gamma$ & 1 & 2 \\
$w_j$ & 2 & 2 & 1 &  $3 - \gamma$ \\ \hline 
\end{tabular}
\end{table}
The sketched alternative approach will select items 1, 2, 3 in the last period $T$, items 1 and 2 in period $T - 1$ and either item 1 or item 2 in the remaining periods. 
The corresponding profit is equal to 
$$(3+2\gamma) + (2+2\gamma) + (1+\gamma)(T-2) = (1+\gamma) T + 3 +2\gamma.$$
The optimal solution instead consists in packing items 3 and 4 in all periods, hence $z^{*} = 3T$. 
Consequently, the approximation ratio of the algorithm is strictly  less than $\frac{1}{2}$ for large enough $T$ 
and cannot be greater than $\frac{1}{3}$ as the number of periods $T$ increases.

\section{An approximation algorithm for a weight constrained IKP variant with two periods
\label{sec:RestrIKP23}}

In this section we consider $\mbox{IKP}^\prime$ for the special cases with two periods only, namely $T=2$. It turns out that a more careful analysis of the solution structures of the knapsack subproblems $KP_t$ by means of an LP model gives much better approximation ratios than the more general results presented in Section~\ref{sec:RestrIKP}. 
Let us define for the optimal solution sets \KPone and \KPsec the following subsets illustrated in Figure~\ref{FigureSi}: 
\begin{itemize}
\item $S_{12}$: subset of items included in both optimal solutions \KPone and \KPsec;
\item $S_1$: remaining subset of items in \KPone; 
\item $S_2^a$: remaining subset of items not exceeding capacity $c_1$ in \KPsec;
\item $S_2^\prime$: first item exceeding $c_1$ in \KPsec;
\item $S_2^b$: remaining subset of items in \KPsec;
\end{itemize}
Each subset could as well be empty. The dashed lines in Figure \ref{FigureSi} refer to the item $S_2^\prime$ which exceeds the first capacity value. 

\begin{figure}[h]
\centering
\small
\begin{tikzpicture} 
\coordinate (c) at (0.6,0);

\draw  ($(c)+(0,1)$) rectangle  ($(c)+(2,1.5)$);
\node at  ($(c)+(1,1.25)$) {$S_{12}$}; 

\draw  ($(c)+(2,1)$) rectangle ($(c)+(5,1.5)$);
\node at  ($(c)+(3.5,1.25)$) {$S_1$}; 

\draw  ($(c)+(2,0)$) rectangle ($(c)+(4.5,0.5)$);
\node at ($(c)+ (3.25,0.25)$) {$S_2^a$}; 

\draw [dashed]  ($(c)+(4.5,0)$) rectangle ($(c)+(5.5,0.5)$);
\node at  ($(c)+(5,0.25)$) {$S_2^\prime$}; 

\draw  ($(c)+(5.5,0)$) rectangle ($(c)+(7,0.5)$);
\node at  ($(c)+(6.25,0.25)$) {$S_2^b$}; 

\draw  ($(c)+(0,0)$) rectangle ($(c)+(2,0.5)$);
\node at  ($(c)+(1,0.25)$) {$S_{12}$};

\node at (0,1.25) {$KP_1:$};
\node at (0,0.25) {$KP_2:$};

\node at  ($(c)+(5.0,1.75)$) {$c_1$};
\node at  ($(c)+(7.0,0.75)$) {$c_2$};
\end{tikzpicture}
\caption{Decomposition of the optimal solutions of $KP_1$ and $KP_2$.} \label{FigureSi} 
\end{figure}
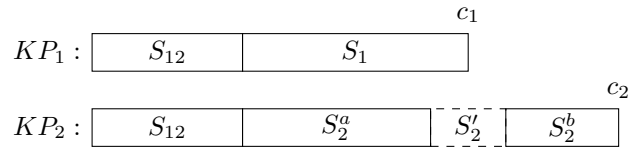

According to the above definitions, we have the following inequalities:
\begin{align}
& w(S_{12}) + w(S_1) \leq c_1 \label{CaseW1} \\
& w(S_{12}) + w(S_2^a) \leq c_1 \label{CaseW2} \\
& w(S_{12}) + w(S_2^a) + w(S_2^\prime) > c_1 \label{CaseW3}\\
& w(S_{12}) + w(S_2^a) + w(S_2^\prime) + w(S_2^b) \leq c_2 \label{CaseW4}\\
& w(S_{12}) + w(S_1) + w(S_2^b) < c_2 \label{CaseW5}
\end{align}
Inequality (\ref{CaseW5}) derives directly from inequalities (\ref{CaseW1}), (\ref{CaseW3}) and (\ref{CaseW4}).
The optimal solution values of $KP_1$ and $KP_2$ are 
$z_1 = p(S_{12}) + p(S_1)$ and  $z_2 = p(S_{12}) + p(S_2^a) + p(S_2^\prime) + p(S_2^b)$ respectively. 
Now we can state three feasible solutions for $\mbox{IKP}^\prime$: 
\begin{enumerate}[label=\alph*)]
\item \label{1Sol} \KPone 
in the two periods plus the additional packing of items in $S_2^b$ in the second period with total profit $$\Delta_1(p(S_{12}) + p(S_1)) + \Delta_2(p(S_{12}) + p(S_1) + p(S_2^b));$$
\item \label{2Sol} \KPsec 
in the second period with the packing of items in subsets $S_{12}$ and $S_2^a$ in the first period and resulting profit $$\Delta_1(p(S_{12}) + p(S_2^a)) + \Delta_2(p(S_{12}) + p(S_2^a)+p(S_2^\prime) + p(S_2^b));$$ 
\item \label{3Sol} \KPsec 
in the second period with item $S_2^\prime$ placed in the knapsack in the first period. The profit of this solution is $$\Delta_1p(S_2^\prime) + \Delta_2(p(S_{12}) + p(S_2^a)+ p(S_2^\prime) + p(S_2^b)).$$
\end{enumerate}
We introduce an algorithm, hereafter denoted as $H_{T2}$, which simply returns the best of these solutions. 
\begin{algorithm}[H]
\begin{algorithmic}[1]
\STATEx \textbf{Input:} $\mbox{IKP}^\prime$ instance with $T=2$.
\STATE Compute \KPone and \KPsec. 
\STATE Identify subsets $S_{12}, S_1, S_2^a, S_2^\prime, S_2^b$ and compute solutions \ref{1Sol}, \ref{2Sol}, \ref{3Sol}.
\STATE Return the best solution found.
\end{algorithmic}
\caption{\textbf{Algorithm $H_{T2}$}}
\end{algorithm}

\begin{theorem}
\label{Theo:T=2}
For $\mbox{IKP}^\prime$ with $T=2$, algorithm $H_{T2}$ has a tight approximation ratio of $\frac{1 + 3\Delta_r + 2\Delta_r^2}{1 + 4\Delta_r + 2\Delta_r^2}$, where $\Delta_r = \frac{\Delta_2}{\Delta_1}$.
\end{theorem}

\begin{proof}
In order to evaluate the worst case performance of algorithm $H_{T2}$, we consider an LP formulation where we associate a non--negative variable $h$ with the solution value computed by the algorithm. 
In addition, the profits of the subsets $p(S_{(\cdot)})$ are associated with non--negative variables $s_{(\cdot)}$. 
As in Section \ref{sec:HarmNumRatios}, the positive parameter $OPT$ represents $z^*$. This implies the following parametric LP model with $\Delta_1$ and $\Delta_2$:
\begin{align}
   \text{minimize}\quad & h \label{eq:ObjT2}\\
	\text{subject to}\quad
	& h - (( \Delta_1 +  \Delta_2)(s_{12} + s_1) + \Delta_2 s_2^b)\geq 0 \label{eq:AlgoH1}\\
	& h - (( \Delta_1 +  \Delta_2)(s_{12} + s_2^a) + \Delta_2(s_2^\prime + s_2^b)) \geq 0 \label{eq:AlgoH2}\\
	& h - (( \Delta_1 +  \Delta_2)s_2^\prime + \Delta_2( s_{12} + s_2^a + s_2^b))  \geq 0  \label{eq:AlgoH3}\\
	& \Delta_1(s_{12} + s_1) + \Delta_2(s_{12} + s_2^a + s_2^\prime + s_2^b) \geq OPT \label{eq:UBAlgoH}\\
	& h, s_{12}, s_1, s_2^a, s_2^\prime, s_2^b \geq 0 \label{eq:AlgoHprof}
\end{align}
The objective function value (\ref{eq:ObjT2}) represents a lower bound on the 
worst case performance of algorithm $H_{T2}$. 
Constraints (\ref{eq:AlgoH1})--(\ref{eq:AlgoH3}) guarantee that $H_{T2}$ will select the best of the three feasible solutions. 
Constraint (\ref{eq:UBAlgoH}) states that the sum $\Delta_1 z_1 + \Delta_2 z_2$ constitutes an upper bound on $z^*$. 
Constraints (\ref{eq:AlgoHprof}) 
 indicate that the variables are non-negative. A feasible solution of model (\ref{eq:ObjT2})--(\ref{eq:AlgoHprof}), for any $\Delta_1 > 0$ and $\Delta_2 > 0$, is: 
\begin{align}
& h= \frac{(1 + 3\Delta_r + 2\Delta_r^2)OPT}{1 + 4\Delta_r + 2\Delta_r^2},\, s_1 = \frac{\Delta_r OPT}{\Delta_1(1 + 4\Delta_r + 2\Delta_r^2)} \nonumber \\
& s_{12} = s_2^\prime= \frac{(1 + \Delta_r) OPT}{\Delta_1(1 + 4\Delta_r + 2\Delta_r^2)},\, s_2^a = s_2^b = 0. \nonumber
\end{align}
We will show by strong duality that such a solution is also optimal for any positive value of $\Delta_1$ and $\Delta_2$. If we denote by $\lambda_i$ with $i=1, \dots, 4$ the dual variables related to constraints (\ref{eq:AlgoH1})--(\ref{eq:UBAlgoH}), the dual linear problem corresponding to model (\ref{eq:ObjT2})--(\ref{eq:AlgoHprof}) is as follows:
\begin{align}
   \text{maximize}\quad & OPT \cdot \lambda_4 \label{eq:ObjDualT2}\\
	\text{subject to}\quad
	& \lambda_1 + \lambda_2 + \lambda_3 \leq 1 \label{eq:Dualh}\\
	& \Delta_1\lambda_4 - (\Delta_1 + \Delta_2)\lambda_1 \leq 0 \label{eq:Duals1}\\
	& -\Delta_2\lambda_3 - (\Delta_1 + \Delta_2)(\lambda_1 + \lambda_2 - \lambda_4) \leq 0 \label{eq:Duals12}\\
	& \Delta_2(\lambda_4 - \lambda_3) - (\Delta_1 + \Delta_2)\lambda_2 \leq 0 \label{eq:Duals2a}\\
	& \Delta_2(\lambda_4 - \lambda_2) - (\Delta_1 + \Delta_2)\lambda_3 \leq 0 \label{eq:Duals2prime}\\
	& - \Delta_2(\lambda_1 + \lambda_2 + \lambda_3 - \lambda_4) \leq 0 \label{eq:Duals2b}\\
	& \lambda_1, \lambda_2, \lambda_3, \lambda_4 \geq 0 \label{eq:DualvarT2}
\end{align}
Constraints (\ref{eq:Dualh})--(\ref{eq:Duals2b}) correspond to primal variables $h, s_1, s_{12}, s_2^a, s_2^\prime, s_2^b$ respectively.
A feasible dual solution (for any $\Delta_1 > 0$ and $\Delta_2 > 0$) reads 
\begin{align}
& \lambda_1 = \frac{1 + 2\Delta_r}{1 + 4\Delta_r + 2\Delta_r^2},\, \lambda_2 = \lambda_3 = \frac{(1 + \Delta_r) \Delta_r}{1 + 4\Delta_r + 2\Delta_r^2},\, \lambda_4 = \frac{1 + 3\Delta_r + 2\Delta_r^2}{1 + 4\Delta_r + 2\Delta_r^2}. \nonumber 
\end{align}
Hence, the dual solution value $OPT\cdot \lambda_4 = h$ proves by strong duality the optimality of the primal solution for any value of  the two time multipliers. The corresponding lower bound on the performance ratio provided by algorithm $H_{T2}$ is equal to $\frac{h}{OPT} =  \frac{1 + 3\Delta_r + 2\Delta_r^2}{1 + 4\Delta_r + 2\Delta_r^2}$. 

\smallskip
We can show the tightness of the bound by considering two different instances for $\Delta_r \leq 1$ and $\Delta_r > 1$ respectively, with $n=5$, $c_1 = 3 + \gamma$, $c_2 = 4$ 
and following entries: 
\begin{table}[H]
	\centering
\begin{tabular}{|l|l|ccccc|}
  \hline
& $i$ & 1 & 2 & 3 & 4 & 5 \\ \hline
$\Delta_r \leq 1$ & $p_i$ & $\frac{1+\Delta_r}{\Delta_1}$ & $\frac{1+\Delta_r}{\Delta_1}$ & $\frac{\Delta_r}{\Delta_1}$ & $\frac{1+2\Delta_r -\gamma}{\Delta_1}$ & $\frac{1}{\Delta_1}$ \\
& $w_i$ & $2$ & $2$ & $1 + \gamma$ &  $3 - 2\gamma$ &  $1 + 2\gamma$ \\ \hline \hline
$\Delta_r > 1$ & $p_i$ & $\frac{1+2\Delta_r}{\Delta_1}$ & $\frac{1+\Delta_r}{\Delta_1}$ & $\frac{1+ \Delta_r}{\Delta_1}$ & $\frac{\Delta_r -\gamma}{\Delta_1}$ & $\frac{1}{\Delta_1}$ \\
& $w_i$ & $3 + \gamma$ & $2$ & $2$ &  $1$ &  $1$ \\ \hline
\end{tabular}
\end{table}
First, consider the instance for the case $\Delta_r \leq 1$. The optimal solutions of $KP_1$ and $KP_2$ consist of items 1, 3 and items 1, 2 respectively. Correspondingly, we have $$S_{12} = \{ 1 \}, S_1 = \{ 3 \}, S_2^a = \{ \emptyset \}, S_2^\prime = \{ 2 \}, S_2^b = \{ \emptyset \}.$$
An easy computation reveals that all the solutions provided by algorithm $H_{T2}$ have the same profit value of $1+ 3\Delta_r + 2\Delta_r^2$.
The optimal $\mbox{IKP}^\prime$ solution selects item 4 in the first period together with item 5 in the second one, hence $z^* = 1+ 4\Delta_r + 2\Delta_r^2 - (1 + \Delta_r)\gamma$. \\
In the second instance with $\Delta_r > 1$, we have $Sol(KP_1) = \{ 1 \}, Sol(KP_2) = \{ 2, 3 \}$ and thus $$S_{12} = \{ \emptyset \}, S_1 = \{ 1 \}, S_2^a = \{ 2 \}, S_2^\prime = \{ 3 \}, S_2^b = \{ \emptyset \}.$$
Also in this case all solutions provided by the algorithm have a profit equal to $1+ 3\Delta_r + 2\Delta_r^2$. An optimal solution packs items 2 and 4 in the first period together with item 5 in the second one, hence $z^* = 1+ 4\Delta_r + 2\Delta_r^2 - (1 + \Delta_r)\gamma$. \\
In both cases, the approximation ratio of algorithm $H_{T2}$ can be arbitrarily close to ratio $\frac{1 + 3\Delta_r + 2\Delta_r^2}{1 + 4\Delta_r + 2\Delta_r^2}$ as the value of $\gamma$ goes to 0.
\end{proof}

\begin{coro}
\label{Coro:T2_1}
The approximation ratio of algorithm $H_{T2}$ is bounded from below by $\frac{2\sqrt{2} + 3}{2\sqrt{2} + 4}= \frac{1}{2} + \frac{\sqrt{2}}{4}$. When $H_{T2}$ is applied to $\mbox{IIKP}^\prime$,
the ratio is $\frac{6}{7}$.
\end{coro}
\begin{proof}
A straightforward computation reveals that function $\frac{1 + 3\Delta_r + 2\Delta_r^2}{1 + 4\Delta_r + 2\Delta_r^2}$  has a unique global minimum (for $\Delta_r > 0$) at  $\Delta_r = \frac{\sqrt{2}}{2}$. Hence, we get that the ratio $\frac{h}{OPT}$ is bounded from below by $\frac{2\sqrt{2} + 3}{2\sqrt{2} + 4}$. For $\mbox{IIKP}^\prime$, we have $\Delta_r =1$. Correspondingly,
$\frac{1+3\Delta_r+2\Delta^2_r}{1+4\Delta_r+2\Delta^2_r}
=\frac{6}{7}$.
\end{proof}

Considering algorithm $H_{T2}$ the following improvement comes to mind:
Solve to optimality $KP_2$ after packing the item set \KPone in both periods
(as in $H_2$) and, as an alternative, optimally solve $KP_1$ restricted to the item set \KPsec. 
However, the above tight instance would also apply to this computationally more demanding variation and thus no improvement of the approximation ratio can be gained.

Algorithm $H_{T2}$ is not a polynomial time algorithm since it requires the optimal solutions of $KP_1$ and $KP_2$. 
We could solve these knapsack problems by an $\eps$--approximation scheme (FPTAS) to get a polynomial running time at the cost of a decrease of the approximation ratio. 
The following corollary shows that the ratio is decreased by a factor $(1 - \eps)$.
\begin{coro}
\label{Coro:T2}
If an $\eps$--approximation scheme is employed for solving the standard knapsack problems $KP_1$ and $KP_2$, the approximation ratio of algorithm $H_{T2}$ is at least $ \frac{1 + 3\Delta_r + 2\Delta_r^2}{1 + 4\Delta_r + 2\Delta_r^2}(1-\eps)$.
\end{coro}

\begin{proof}
We consider again model (\ref{eq:ObjT2})--(\ref{eq:AlgoHprof}) to evaluate the performance of the algorithm. If the subsets $S_{(\cdot)}$ 
correspond to the items of the $\eps$--approximations for $KP_1$ and $KP_2$, constraints (\ref{eq:AlgoH1})--(\ref{eq:AlgoH3}) straightforwardly hold.
Then, we just replace constraint (\ref{eq:UBAlgoH}) by constraint
\begin{equation}
\label{eq:epsUBAlgoH}
 \Delta_1(s_{12} + s_1) + \Delta_2(s_{12} + s_2^a + s_2^\prime + s_2^b) \geq OPT(1-\eps)
\end{equation}
which indicates that the weighted sum of the approximate solutions divided by $(1-\eps)$ provides an upper bound on the optimal solution value of $\mbox{IKP}$. In the related primal/dual LP models, the optimal values of primal variables are now multiplied by $(1-\eps)$ while the optimal dual solution remains the same with objective value $(1-\eps)OPT \cdot \lambda_4$. The corresponding value of the ratio $\frac{h}{OPT}$ shows the bound. 
\end{proof}

\section{Conclusions}
\label{TheConcl}
We proposed for $\mbox{IKP}$ a series of results extending in different directions the contributions currently available in the literature. We proved the tightness of the approximation ratio of a general purpose algorithm presented in \cite{HaSh06} and originally applied to the time-invariant version of the problem. We also established a polynomial time approximation scheme (PTAS) when one of the problem inputs can be considered as a constant. We then focused on a restricted relevant variant of $\mbox{IKP}$ which plausibly assumes the possible packing of any item since the first period. We discussed the performance of different approximation algorithms for the general case as well as for the variation with two time periods. In future research, we will investigate extensions of our procedures to the design of improving approximation algorithms for variants involving more than two periods. Also, since to the authors' knowledge no computational experience has been provided for $\mbox{IKP}$ so far, it would also be interesting to derive new solution approaches and test their performance after generating benchmarks and challenging to solve instances. 

\subsection*{Acknowledgments}

Ulrich Pferschy was supported by the project ``Choice-Selection-Decision" and by the COLIBRI Initiative of the University of Graz.

 \bibliographystyle{elsarticle-harv} 
 \bibliography{ref8}
 \label{sec:bib}

\end{document}